\documentclass[a4paper,UKenglish]{lipics_arxiv}
\usepackage{microtype}%if unwanted, comment out or use option "draft"
\usepackage{comment}

\title{Agglomerative Clustering of Growing Squares\footnote{The final publication is available at Springer via \url{https://doi.org/10.1007/978-3-319-77404-6_20}. The Netherlands Organisation for Scientific Research (NWO) is supporting B.S. under project no.~639.023.208, K.V. under project no.~639.021.541, and T.C. under project no.~314.99.117. The Danish National Research Foundation is supporting F.S. under grant nr.~DNRF84.}}
\titlerunning{Agglomerative Clustering of Growing Squares} %optional, in case that the title is too long; the running title should fit into the top page column

\author[1]{Thom~Castermans}
\author[2]{Bettina~Speckmann}
\author[3]{Frank~Staals}
\author[4]{Kevin~Verbeek}
\affil[1]{Department of Computer Science, TU Eindhoven, Eindhoven, Netherlands\\
  \texttt{t.h.a.castermans@tue.nl}}
\affil[2]{Department of Computer Science, TU Eindhoven, Eindhoven, Netherlands\\
  \texttt{b.speckmann@tue.nl}}
\affil[3]{MADALGO, Aarhus University, Aarhus, Denmark\\
  \texttt{f.staals@cs.au.dk}}
\affil[4]{Department of Computer Science, TU Eindhoven, Eindhoven, Netherlands\\
  \texttt{k.a.b.verbeek@tue.nl}}
\authorrunning{T. Castermans, B. Speckmann, F. Staals, and K. Verbeek} %mandatory. First: Use abbreviated first/middle names. Second (only in severe cases): Use first author plus 'et. al.'

\Copyright{Thom Castermans, Bettina Speckmann, Frank Staals, and Kevin Verbeek}%mandatory, please use full first names. LIPIcs license is "CC-BY";  http://creativecommons.org/licenses/by/3.0/

% mandatory: Please choose ACM 1998 classifications from http://www.acm.org/about/class/ccs98-html . E.g., cite as "F.1.1 Models of Computation".
% mandatory: Please provide 1-5 keywords
\usepackage{slantsc}
\AtBeginDocument{% https://tex.stackexchange.com/a/284343
  \DeclareFontShape{T1}{lmr}{m}{scit}{<->ssub*lmr/m/scsl}{}%
}
\usepackage{wrapfig}
\usepackage{xspace}

\graphicspath{{../figures/}}

\DeclareMathOperator{\polylog}{polylog}

\newcommand*{\etal}{et al.\xspace}
\newcommand*{\BBalpha}{\textsc{bb}{\footnotesize[}$\alpha${\footnotesize]}}
\newcommand*{\R}{\mathbb{R}}
\newcommand*{\numLinks}[1]{\textit{numLinks}\,(#1)}

\newcommand{\mkmcal}[1]{\ensuremath{\mathcal{#1}}\xspace}
\newcommand{\X}{\mkmcal{X}}

\theoremstyle{definition} % https://tex.stackexchange.com/a/73602
\newtheorem{observation}[theorem]{Observation}

\begin{document}

\maketitle

\begin{abstract}
  We study an agglomerative clustering problem motivated by interactive glyphs
  in geo-visualization. Consider a set of disjoint square glyphs on an
  interactive map. When the user zooms out, the glyphs grow in size relative to
  the map, possibly with different speeds. When two glyphs intersect, we wish
  to replace them by a new glyph that captures the information of the
  intersecting glyphs.

  We present a fully dynamic kinetic data structure that maintains a set of $n$
  disjoint growing squares. Our data structure uses
  $O(n (\log n \log\log n)^2)$ space, supports queries in worst case
  $O(\log^3 n)$ time, and updates in $O(\log^7 n)$ amortized time. This leads
  to an $O(n\alpha(n)\log^7 n)$ time algorithm to solve the agglomerative
  clustering problem. This is a significant improvement over the current best
  $O(n^2)$ time algorithms.
\end{abstract}

\section{Introduction}
\label{sec:Introduction}

We study an agglomerative clustering problem motivated by interactive glyphs in geo-visualization. Our specific use case stems from the eHumanities, but similar visualizations are used in a variety of application areas. \emph{GlamMap}~\cite{csvwbb-ggfeh-2016}\footnote{{\tt http://glammap.net/glamdev/maps/1}, best viewed in Chrome.} is a visual analytics tool which allows the user to interactively explore datasets which contain (at least) the following metadata of a book collection: author, title, publisher, year of publication, and location (city) of publisher. Each book is depicted by a square, color-coded by publication year, and placed on a map according to the location of its publisher. Overlapping squares (many books are published in Leipzig, for example) are recursively aggregated into a larger glyph until all glyphs are disjoint (see Fig.~\ref{fig:glammap}). As the user zooms out, the glyphs ``grow'' relative to the map to remain legible. As a result, glyphs start to overlap and need to be merged into larger glyphs to keep the map clear and uncluttered. It is straightforward to compute the resulting agglomerative clustering whenever a data set is loaded and to serve it to the user as needed by the current zoom level. However, GlamMap allows the user to filter by author, title, year of publication, or other applicable meta data. It is impossible to pre-compute the clustering for any conceivable combination of filter values. To allow the user to browse at interactive speeds, we hence need an efficient agglomerative clustering algorithm for growing squares (glyphs). Interesting bibliographic data sets (such as the catalogue of WorldCat, which contains more than 321 million library records at hundreds of thousands of distinct locations) are too large by a significant margin to be clustered fast enough with the current state-of-the-art $O(n^2)$ time algorithms (here $n$ is the number of squares or glyphs).

In this paper we formally analyze the problem and present a fully dynamic data
structure that uses $O(n (\log n \log\log n)^2)$ space, supports updates in $O(\log^7 n)$
amortized time, and queries in $O(\log^3 n)$ time, which allows us to compute
the agglomerative clustering for $n$ glyphs in $O(n\alpha(n)\log^7 n)$ time. Here,
 $\alpha$ is the extremely slowly growing inverse Ackermann function.  To
the best of our knowledge, this is the first fully dynamic clustering algorithm
which beats the classic $O(n^2)$ time bound.

\begin{figure}[t]
  \includegraphics[width=0.3\linewidth]{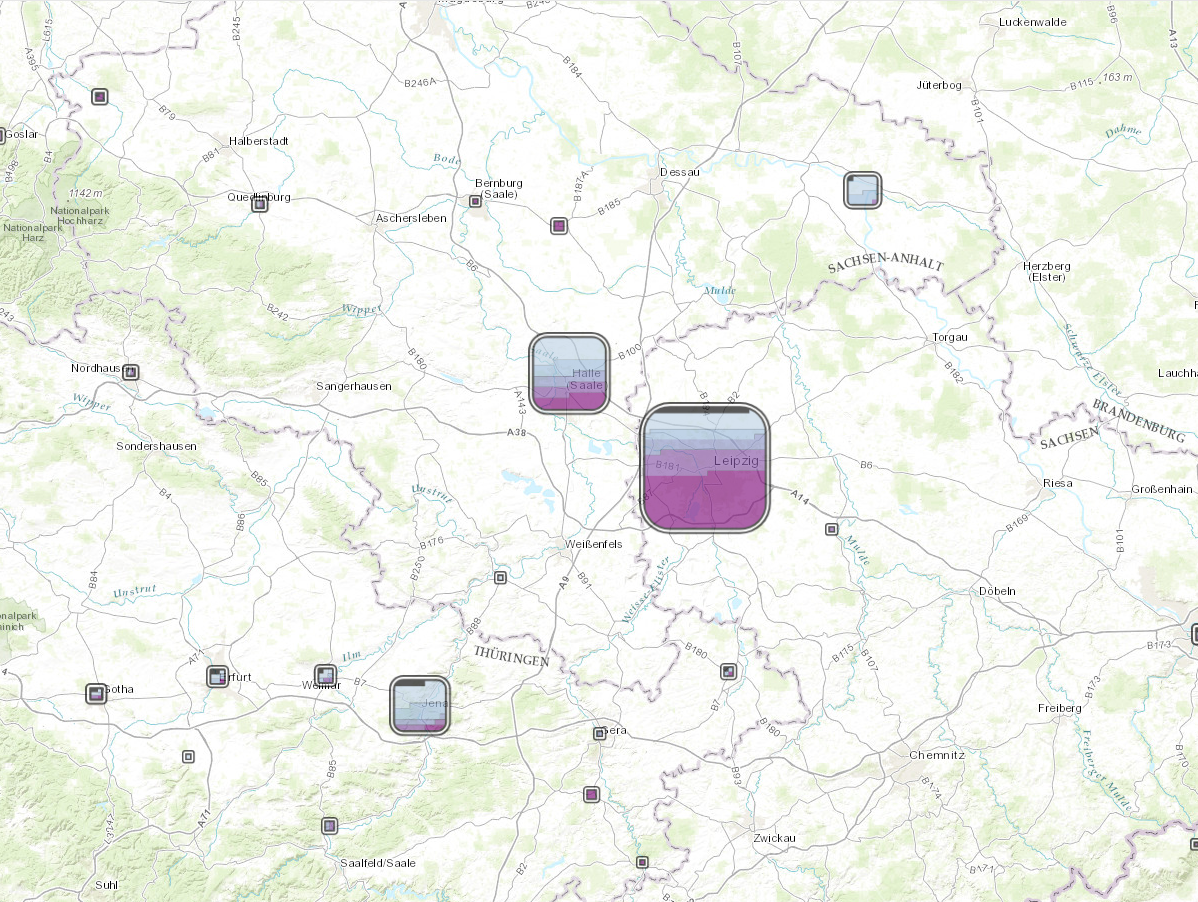}\hfill
  \includegraphics[width=0.3\linewidth]{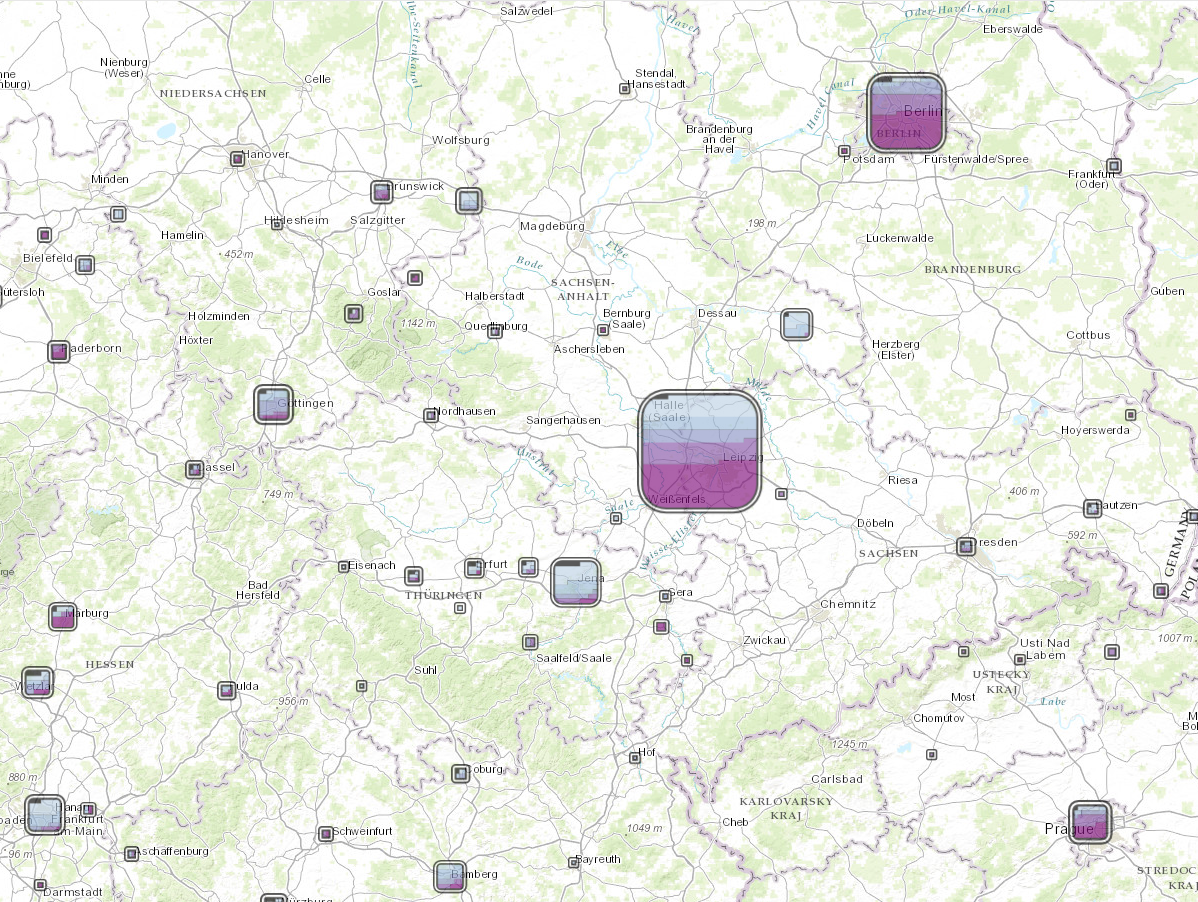}\hfill
  \includegraphics[width=0.3\linewidth]{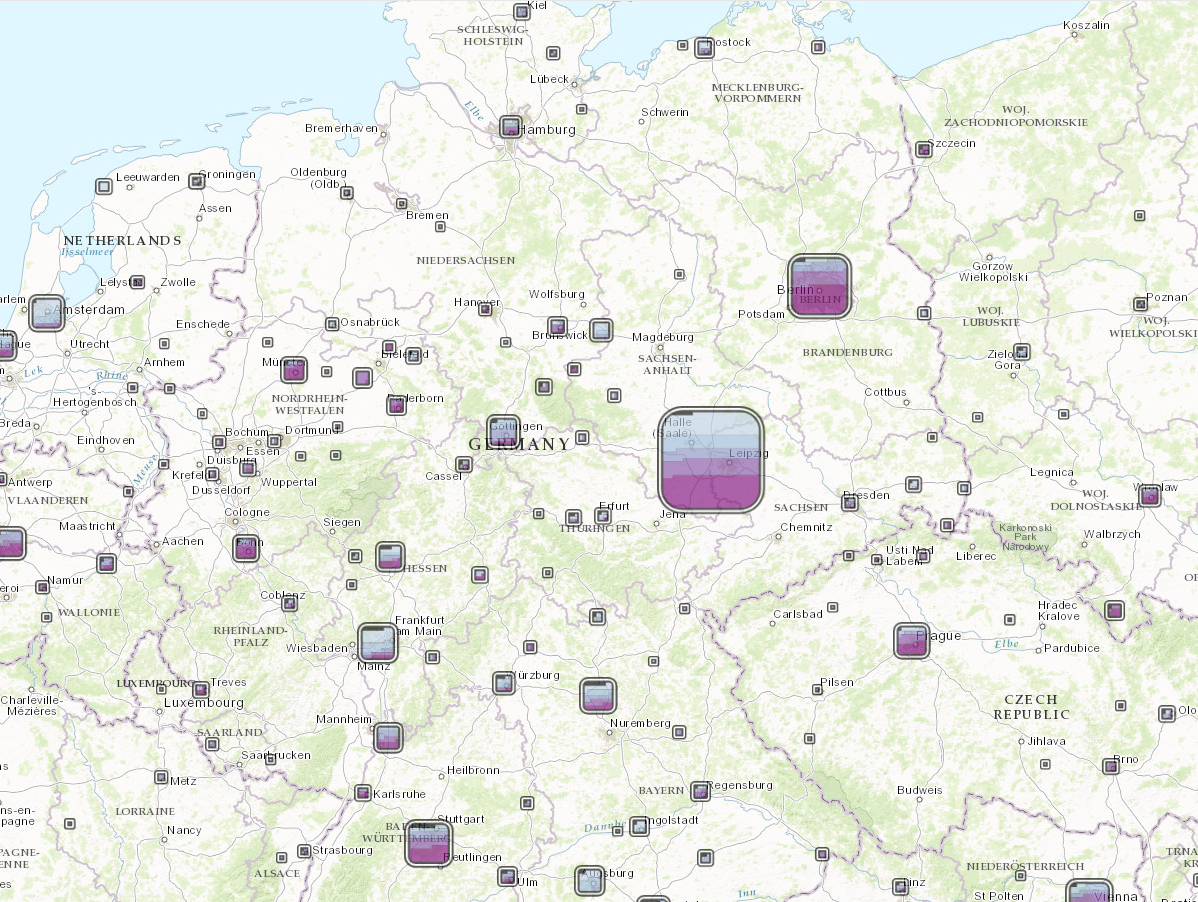}
  \caption{Zooming out in GlamMap will merge overlapping squares. This figure
  shows a sequence of three steps zooming out from the surroundings of Leipzig.}
  \label{fig:glammap}
\end{figure}

\begin{figure}[b]
  \centering\includegraphics{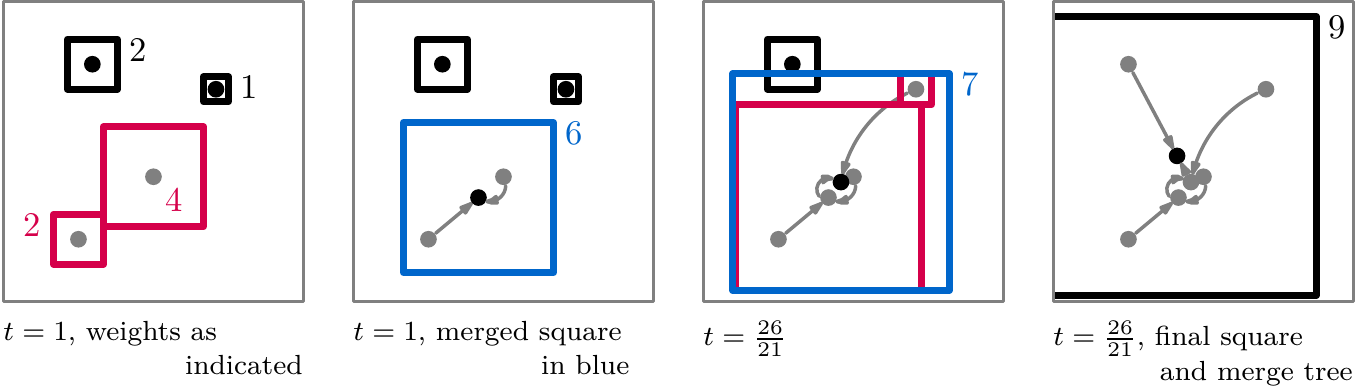}
  \caption{The timeline of squares that grow and merge as they touch.}
  \label{fig:zoom-out}
\end{figure}

\subparagraph{Formal problem statement.} Let $P$ be a set of points in $\R^2$ (the locations of publishers from our example). Each point $p \in P$
has a positive weight $p_w$ (number of books published in this city). Given a ``time'' parameter $t$, we interpret the
points in $P$ as squares. More specifically, let $\square_p(t)$ be the square
centered at $p$ with width $tp_w$. For ease of exposition we assume all $x$ and $y$ to be unique.
With some abuse of notation we may refer to
$P$ as a set of squares rather than the set of center points of
squares. Observe that initially, i.e. at $t = 0$, all squares in $P$ are
disjoint. As $t$ increases, the squares in $P$ grow, and hence they may start
to intersect. When two squares $\square_p(t)$ and $\square_q(t)$
intersect at time $t$, we remove both $p$ and $q$ and replace them by a new
point $z = \alpha p + (1-\alpha)q$, with $\alpha = w_p/(w_p+w_q)$, of weight
$z_w = p_w + q_w$ (see Fig.~\ref{fig:zoom-out}).

%\frank{we will write $\square_p$ as the general square centered at $p$}
%Define $a(t)$ and $A(t)$, for objects $a$ and sets of objects $A$.

\subparagraph{Related Work.} Funke, Krumpe, and
Storandt~\cite{funke2016crushing} introduced so-called ``ball tournaments'', a
related, but simpler, problem, which is motivated by map labeling. Their input
is a set of balls in $\mathbb{R}^d$ with an associated set of priorities. The
balls grow linearly and whenever two balls touch, the ball with the lower
priority is eliminated. The goal is to compute the elimination sequence
efficiently. Bahrdt~\etal~\cite{bahrdt2017growing} and Funke and
Storandt~\cite{funke2017parametrized} improved upon the initial results and
presented bounds which depend on the ratio $\Delta$ of the largest to the
smallest radius. Specifically, Funke and Storandt~\cite{funke2017parametrized}
show how to compute an elimination sequence for $n$ balls in
$O(n \log \Delta (\log + \Delta^{d-1}))$ time in arbitrary dimensions and in
$O(Cn \polylog n)$ time for $d=2$, where $C$ denotes the number of different
radii.
In our setting eliminations are not sufficient, since merged glyphs
need to be re-inserted. Furthermore, as opposed to typical map labeling
problems where labels come in a fixed range of sizes, the sizes of our glyphs
can vary by a factor of 10.000 or more (Amsterdam with its many
well-established publishers vs. Kaldenkirchen with one obscure one).

Ahn~\etal~\cite{ahn2017faster} very recently and independently developed the first sub-quadratic algorithms to compute elimination orders for ball tournaments. Their results apply to balls and boxes in two or higher dimensions. Specifically, for squares in two dimensions they can compute an elimination order in $O(n \log^4 n)$ time. Their results critically depend on the fact that they know the elimination priorities at the start of their algorithm and that they only have to handle deletions. Hence they do not have to run an explicit simulation of the growth process and can achieve their results by the clever use of advanced data structures. In contrast, we are handling the fully dynamic setting with both insertions and deletions, and without a specified set of priorities.
%Furthermore, our data structures are comparatively simple, which makes it possible to implement them in a visual analytics system.
%\frank{I'm not entirely sure our approach is simpler to implement than
%  theirs. For the squares case they also seem to use fairly simple data structures}

Our clustering problem combines both dynamic and kinetic aspects: squares grow, which is a restricted form of movement, and squares are both inserted and deleted. There are comparatively few papers which tackle dynamic kinetic problems.  Alexandron
\etal~\cite{alexandron2007kinetic_envelope} present a dynamic and kinetic data
structure for maintaining the convex hull of points (or analogously, the lower
envelope of lines) moving in $\R^2$. Their data structure processes
(in expectation) $O(n^2\beta_{s+2}(n)\log n)$ events in $O(\log^2 n)$ time each.
Here, $\beta_{s}(n) = \lambda_s(n)/n$, and $\lambda_s(n)$ is the maximum length
of a Davenport-Schinzel sequence on $n$ symbols of order $s$.
Agarwal \etal~\cite{agarwal2008kinetictournament} present dynamic and kinetic
data structures for maintaining the closest pair and all nearest neighbors. The
expected number of events processed is again roughly
$O(n^2\beta_{s+2}(n)\polylog n)$, each of which can be handled in
$O(\polylog n)$ expected time. We are using some idea and constructions which
are similar in flavor to the structures presented in their paper.

%Vaguely related in the sense that it uses a simulation of a kinetic DS to
%compute something static Aronov \etal~\cite{agtz-vqmsp-02} (weak visibility
%from a line segment inside a simple polygon). \frank{I guess that there may be
%  many more though}.

%\bettina{I got to here, text after needs a pass}

\subparagraph{Results.} We present a fully dynamic data structure that can
maintain a set $P$ of disjoint growing squares. Our data structure will produce
an \emph{intersection event} at every time $t$ when two squares $\square_p$ and
$\square_q$, with $p,q \in P$, start to intersect (i.e.~at any time before $t$,
all squares in $P$ remain disjoint). At such a time, we then have to delete
some of the squares, to make sure that the squares in $P$ are again
disjoint. At any time, our data structure supports inserting a new square that
is disjoint from the squares in $P$, or removing an existing square from
$P$. Our data structure can handle a sequence of $m \geq n$ updates in a total
of $O(m\alpha(n)\log^7 n)$ time, each update is performed in $O(\log^7 n)$
amortized time.

%\frank{Observe that everything also works for rectangular objects, as long as
%  all of them have the same aspect ratio.}

\subparagraph{The Main Idea.}
We develop a data structure that can maintain a dynamic set of disjoint squares
$P$, and produce an intersection event at every time $t$ when $\square_q$
starts to intersect with a square $\square_p$ of a point $p \in P$ that
\emph{dominates $q$}. We say that a point $p$ dominates $q$ if and only if
$q_x \leq p_x$ and $q_y \leq p_y$. We then combine four of these data
structures, one for each quadrant, to make sure that all squares in $P$ remain
disjoint. The main observation that allows us to maintain $P$ efficiently, is
that we can maintain the points $D(q)$ dominating $q$ in an order so that a
prefix of $D(q)$ will have their squares intersect the top side of $\square_q$
first, and the remaining squares will intersect the right side of $\square_q$
first. We formalize this in Section~\ref{sec:Geometric_Properties}. We then
present our data structure --essentially a pair of range trees interlinked with
\emph{linking certificates}-- in Section~\ref{sec:DataStructure}. While our
data structure is conceptually simple, the exact implementation is somewhat
intricate, and the details are numerous. Our initial analysis shows that our
data structure maintains $O(\log^6 n)$ certificates per square, which yields an
$O(\log^7 n)$ amortized update time. This allows us to simulate the process of
growing the squares in $P$ --and thus solve the agglomerative glyph clustering problem-- in
$O(n\alpha(n)\log^7 n)$ time using $O(n\log^5 n)$ space. In
Section~\ref{sec:Link_Relations} we analyze the relation between canonical
subsets in dominance queries. We show that for two range trees $T^R$ and $T^B$
in $\R^d$, the number of pairs of nodes $r \in T^R$ and $b \in T^B$ for which
$r$ occurs in the canonical subset of a dominance query defined by $b$
and vice versa is only $O(n (\log n \log\log n)^2)$, where $n$ is the total size of $T^R$
and $T^B$. This implies that the number of linking certificates that our data
structure maintains, as well as the total space used, is actually only $O(n (\log n \log\log n)^2)$. Since the linking certificates actually provide an efficient representation of all dominance relations between two point sets (or within a point set), we believe that this result is of independent interest as well.

% but we think it is of
%independent interest as well. These results lead to an improved bound on the
%space used by our data structure. % Unfortunately, we do not yet know how to use
% these results to further improve the running time of our algorithm. -

 % Finally, in
% Section~\ref{sec:Deletions_Only} we present a variant of our data structure for
% the case that we allow only deleting squares from $P$. This simplified version
% allows deletions in $O(\log^5 n)$ time, and processes only $O(....)$ events in
% $O(..)$ time. Hence, this yields a solution to the ball-tournaments problem
% (with squares instead of balls) with the same bounds.

\begin{wrapfigure}[14]{R}{0.5\linewidth}
  \vspace{-1.6cm}
  \centering
  \includegraphics{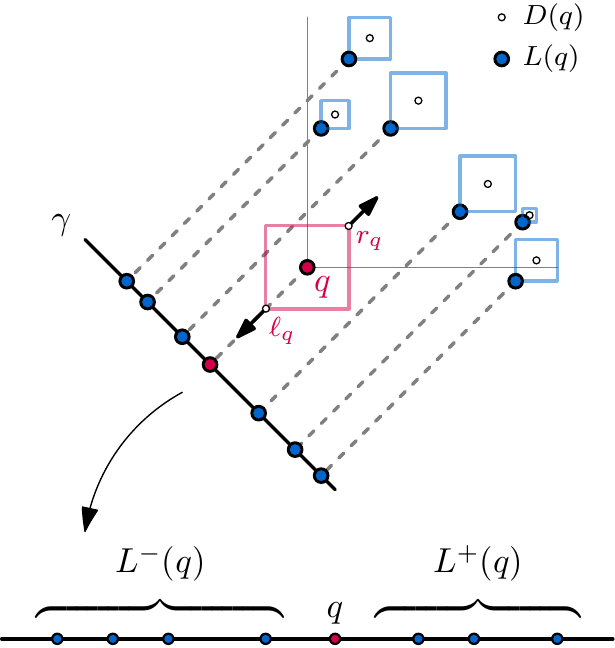}
  \caption{The squares and the projection of their centers and relevant corners
  onto the line $\gamma$.}
  \label{fig:dominates}
  \vspace{-2em} % HACK: tweak to make figure fit
\end{wrapfigure}

\section{Geometric Properties}
\label{sec:Geometric_Properties}

Let $\ell_q$ denote the bottom left vertex of a square $\square_q$, and let
$r_q$ denote the top right vertex of $\square_q$. Furthermore, let $D(q)$
denote the subset of points of $P$ dominating $q$, and let
$L(q) = \{\ell_p \mid p \in D(q) \}$ denote the set of bottom left vertices of
the squares of those points.

\begin{observation}
  \label{obs:dominating}
  Let $p \in D(q)$ be a point dominating point $q$. The squares $\square_q(t)$
  and $\square_p(t)$ intersect at time $t$ if and only if $r_q(t)$ dominates
  $\ell_p(t)$ at time $t$.
\end{observation}

Consider a line $\gamma$ with slope minus one, project all points in
$Z(t)= \{r_q(t)\} \cup L(q)(t)$, for some time $t$, onto $\gamma$, and order
them from left to right. Observe that, since all points in $Z$ move along lines
with slope one, this order does not depend on the time $t$. Moreover, for any
point $p$, we have $r_p(0)=\ell_p(0)=p$, so we can easily compute this order by
projecting the centers of the squares onto $\gamma$ and sorting them. Let
$D^-(q)$ denote the (ordered) subset of $D(q)$ that occur before $q$
in the order along $\gamma$, and let $D^+(q)$ denote the ordered subset of
$D(q)$ that occur after $q$ in the order along $\gamma$. We define $L^-(q)$ and
$L^+(q)$ analogously.

\begin{observation}
  \label{obs:intersect}
  Let $p \in D(q)$ be a point dominating point $q$, and let $t^*$ be the first
  time at which $r=r_q(t^*)$ dominates $\ell=\ell_p(t^*)$. We then have that
  \begin{itemize}
  \item $\ell_x < r_x$ and $\ell_y=r_y$ if and only if $p \in D^-(q)$, and
  \item $\ell_x=r_x$ and $\ell_y < r_y$ if and only if $p \in D^+(q)$.
  \end{itemize}
  See Fig.~\ref{fig:dominates} for an illustration.
\end{observation}

Observation~\ref{obs:intersect} implies that the points $p$ in $D^-(q)$ will
start to intersect $\square_q$ at some time $t^*$ because the bottom left vertex
$\ell_p$ of $\square_p$ will enter $\square_q$ through the top edge, whereas
the bottom left vertex of the (squares of the) points in $D^+(q)$ will enter
$\square_q$ through the right edge. We thus obtain the following result.

\begin{lemma}
  \label{lem:intersect}
  Let $t^*$ be the first time at which a square $\square_p$ of a point
  $p \in D(q)$ intersects $\square_q$. We then have that

  \hspace{-2em}
  \begin{tabular}{b{0.5\linewidth}p{0.01\linewidth}p{0.4\linewidth}}
    \begin{enumerate}
    \item[(i)] $r_q(t^*)_y = \ell_p(t^*)_y$, and $\ell_p(t^*)$ is the point
      with minimum $y$-coordinate among the points in $L^-(q)(t^*)$ at time
      $t^*$,\vspace*{-\baselineskip}
    \end{enumerate}
    && if and only if $p \in D^-(q)$, and \\
    \begin{enumerate}
    \item[(ii)] $r_q(t^*)_x = \ell_p(t^*)_x$, and $\ell_p(t^*)$ is the point
      with minimum $x$-coordinate among the points in $L^+(q)(t^*)$ at time
      $t^*$, \vspace*{-\baselineskip}
    \end{enumerate}
    && \vspace*{-\baselineskip} otherwise (i.e.~if and only if $p \in D^+(q)$).
  \end{tabular}
\end{lemma}

\section{A Kinetic Data Structure for Growing Squares}
\label{sec:DataStructure}

In this section we present a data structure that can detect the first
intersection among a dynamic set of disjoint growing squares. In particular, we
describe a data structure that can detect intersections between all pairs of
squares $\square_p, \square_q$ in $P$ such that $p \in D^+(q)$. We build an
analogous data structure for when $p \in D^-(q)$. This covers all intersections
between pairs of squares $\square_p, \square_q$, where $p \in D(q)$. We then
use four copies of these data structures, one for each quadrant, to detect the
first intersection among all pairs of squares.

We describe the data structure itself in Section~\ref{sub:The_DataStructure},
and we briefly describe how to query it in Section~\ref{sub:Queries}. We deal
with updates, e.g. inserting a new square into $P$ or deleting an existing
square from $P$, in Section~\ref{sub:Updates}. In
Section~\ref{sub:Running_the_Simulation} we analyze the total number of events
that we have to process, and the time required to do so, when we grow the
squares.

\subsection{The Data Structure}
\label{sub:The_DataStructure}

Our data structure consists of two three-layered trees $T^L$ and $T^R$, and
a set of certificates linking nodes in $T^L$ to nodes in $T^R$. These trees
essentially form two 3D range trees on the centers of the squares in $P$, taking
third coordinate $p_\gamma$ of each point to be their rank in the order along
the line $\gamma$ (ordered from left to right). The third layer of $T^L$ will
double as a kinetic tournament tracking the bottom left vertices of squares.
Similarly, $T^R$ will track the top right vertices of the squares.

\subparagraph{The Layered Trees.} The tree $T^L$ is a 3D-range tree storing the
center points in $P$. Each layer is implemented by a weight-balanced binary
search tree (\BBalpha{}~tree)~\cite{nievergelt1973bbalphatree}, and each node
$\mu$ corresponds to a canonical subset $P_\mu$ of points stored in the leaves
of the subtree rooted at $\mu$. The points are ordered on $x$-coordinate first,
then on $y$-coordinate, and finally on $\gamma$-coordinate. Let $L_\mu$ denote
the set of bottom left vertices of squares corresponding to the set $P_\mu$,
for some node $\mu$.

% The tree $T^L$ is a weight-balanced binary
% search tree (\BBalpha{}~tree)~\cite{nievergelt1973bbalphatree} storing the
% center points in $P$ ordered by increasing $x$-coordinate. Each internal node
% $u$ corresponds to a canonical subset $P_u$ of points that are stored in the
% leaves of the subtree rooted at $u$, and stores an associated \BBalpha{}~tree
% $T_u$ containing those points ordered on increasing $y$-coordinate. Every node
% $v$ in this tree $T_u$ again corresponds to a canonical subset
% $P_v \subseteq P_u$ of points stored in the leaves of the subtree rooted at
% $v$. Node $v$ again has an associated \BBalpha{}~tree $X^L_v$ storing the
% points of $P_v$ in its leaves, this time ordered on increasing
% $\gamma$-coordinate. For every internal node $w \in X^L_v$ we again denote its
% canonical subset by $P_w$. Let $L_\kappa$ denote the set of bottom left
% vertices of squares corresponding to the set $P_\kappa$, for some node
% $\kappa$.

Consider the associated structure $X^L_v$ of some secondary node $v$. We
consider $X^L_v$ as a kinetic tournament on the $x$-coordinates of the points
$L_v$~\cite{agarwal2008kinetictournament}. More specifically, every internal
node $w \in X^L_v$ corresponds to a set of points $P_w$ consecutive along the
line $\gamma$. Since the $\gamma$-coordinates of a point $p$ and its bottom left
vertex $\ell_p$ are equal, this means $w$ also corresponds to a set of
consecutive bottom left vertices $L_w$. Node $w$ stores the vertex $\ell_p$ in
$L_w$ with minimum $x$-coordinate, and will maintain certificates that
guarantee this~\cite{agarwal2008kinetictournament}.

The tree $T^R$ has the same structure as $T^L$: it is a three-layered range tree
on the center points in $P$. The difference is that a ternary structure
$X^R_v$, for some secondary node $v$, forms a kinetic tournament maintaining
the maximum $x$-coordinate of the points in $R_v$, where $R_v$ are the top right
vertices of the squares (with center points) in $P_v$. Hence, every ternary
node $z \in X^R_v$ stores the vertex $r_q$ with maximum $x$-coordinate among
$R_v$.

Let $\X^L$ and $\X^R$ denote the set of all kinetic tournament nodes in $T^L$
and $T^R$, respectively.

\subparagraph{Linking the Trees.} Next, we describe how to add \emph{linking
  certificates} between the kinetic tournament nodes in the trees $T^L$ and
$T^R$ that guarantee the squares are disjoint. More specifically, we describe
the certificates, between nodes $w \in \X^L$ and $z \in \X^R$, that
guarantee that the squares $\square_p$ and $\square_q$ are disjoint, for all pairs
$q \in P$ and $p \in D^+(q)$. % Analogously, we add certificates between $w \in
% \Y^L$ and $z \in \Y^R$ that certify that $p$ and $q$, with $p \in D^-(q)$ are
% disjoint.

Consider a point $q$. There are $O(\log^2 n)$ nodes in the secondary trees of
$T^L$, whose canonical subsets together represent exactly $D(q)$. For each of
these nodes $v$ we can then find $O(\log n)$ nodes in $X^L_v$ representing the
points in $L^+(q)$. So, in total $q$ is \textit{interested in} a set $Q^L(q)$
of $O(\log^3 n)$ kinetic tournament nodes. It now follows from
Lemma~\ref{lem:intersect} that if we were to add certificates certifying that
$r_q$ is left of the point stored at the nodes in $Q^L(q)$ we can detect when
$\square_q$ intersects with a square of a point in $D^+(q)$. However,
as there may be many points $q$ interested in a particular kinetic
tournament node $w$, we cannot afford to maintain all of these certificates. The
main idea is to represent all of these points $q$ by a number of canonical subsets
of nodes in $T^R$, and add certificates to only these nodes.

Consider a point $p$. Symmetric to the above construction, there are
$O(\log^3 n)$ nodes in kinetic tournaments associated with $T^R$ that together
exactly represent the (top right corners of) the points $q$ dominated by $p$
and for which $p \in D^+(q)$. Let $Q^R(p)$ denote this set of kinetic
tournament nodes.

\begin{figure}[tb]
  \centering
  \includegraphics{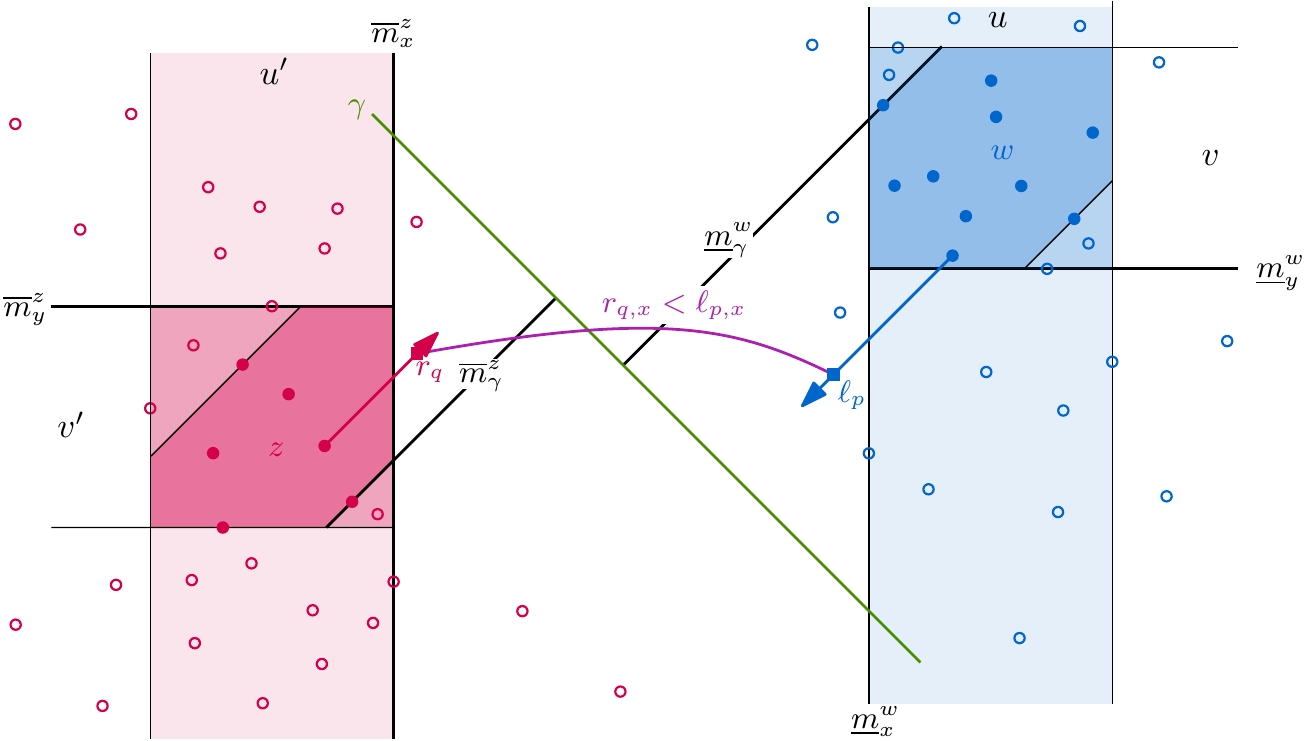}
  \caption{The points $\overline{m}^z$ and $\underline{m}^w$ are defined by a
    pair of nodes $z \in \X^R_{v'}$, with $v' \in T_{u'}$, and $w \in X^L_v$,
    with $v \in T_u$. If $w \in Q^L(\overline{m}^z)$ and $z \in
    Q(\underline{m}^w)$ then we add a linking certificate between the rightmost
    upper right-vertex $r_q$, $q \in P_z$, and the leftmost bottom left vertex
    $\ell_p$, $p \in P_w$.}
  \label{fig:boxes}
\end{figure}

Next, we extend the definitions of $Q^L$ and $Q^R$ to kinetic tournament
nodes. To this end, we first associate each kinetic tournament node with a
(query) point in $\R^3$. Consider a kinetic tournament node $w$ in a tournament
$X^L_v$, and let $u$ be the node in the primary $T^L$ for which $v \in T_u$.
Let
$\underline{m}^w = (\min_{a \in P_u} a_x, \min_{b \in P_v} b_y, \min_{c \in
  P_w} c_\gamma)$ be the point associated with $w$ (note that we take the
minimum over different sets $P_u$, $P_v$, and $P_w$ for the different
coordinates), and define $Q^R(w) = Q^R(\underline{m}^w)$. Symmetrically, for a
node $z$ in a tournament $X^R_v$, with $v \in T_u$ and $u \in T^R$, we define
$\overline{m}^z = (\max_{a \in P_u} a_x, \max_{b \in P_v} b_y, \max_{c \in P_z}
c_\gamma)$ and $Q^L(z) = Q^L(\overline{m}^z)$.

We now add a linking certificate between every pair of nodes $w \in \X^L$ and
$z \in \X^R$ for which (i) $w$ is a node in the canonical subset of $z$, that
is $w \in Q^L(z)$, \emph{and} (ii) $z$ is a node in the canonical subset of
$w$, $z \in Q^R(w)$. Such a certificate will guarantee that the point $r_q$
currently stored at $z$ lies left of the point $\ell_p$ stored at $w$.
\begin{lemma}
  \label{lem:locality}
  Every kinetic tournament node is involved in $O(\log^3 n)$ linking
  certificates, and thus every point $p$ is associated with at most $O(\log^6
  n)$ certificates.
\end{lemma}
\begin{proof}
  We start with the first part of the lemma statement. Every node $w \in \X^L$ can be
  associated with at most $O(\log^3 n)$ linking certificates: one with each
  node in $Q^R(w)$. Analogously, every node $z \in \X^R$ can be associated with
  at most $O(\log^3 n)$ linking certificates: one for each node in $Q^L(z)$.

  Every point $p$ occurs in the canonical subset of at most $O(\log^3 n)$
  kinetic tournament nodes in both $\X^L$ and $\X^R$: $p$ is stored in
  $O(\log^2 n)$ leaves of the kinetic tournaments, and in each such a
  tournament it can participate in $O(\log n)$ certificates (at most two
  tournament certificates in $O(\log n)$ nodes). As we argued above, each such
  a node itself occurs in at most $O(\log^3 n)$ certificates. The lemma
  follows.
\end{proof}
%
\begin{comment}
\begin{proof}
  Every node $w \in \X^L$ can be associated with at most $O(\log^3 n)$
  linking certificates: one with each node in $Q^R(w)$. Analogously, every node
  $z \in \X^R$ can be associated with at most $O(\log^3 n)$
  linking certificates: one for each node in $Q^L(z)$.
\end{proof}

\begin{lemma}
  \label{lem:locality}
  Every center point $p \in P$ is associated with $O(\log^6 n)$ certificates.
\end{lemma}

\begin{proof}
  Every point $p$ occurs in the canonical subset of $O(\log^6 n)$ kinetic
  tournament nodes in both $\X^L$ and $\X^R$. It occurs in at most two
  tournament certificates in each such a node. By
  Lemma~\ref{lem:locality_linking_certificates}, each such a node occurs in at
  most $O(\log^3 n)$ certificates.
\end{proof}
\end{comment}
%
What remains to argue is that we can still detect the first upcoming
intersection. % To this end, we first provide and prove
% Lemma~\ref{lem:matching_pairs}, that is then used to argue, in
% Lemma~\ref{lem:detect_intersecting_squares}, that we can detect the first
% upcoming intersection.

\begin{lemma}
  \label{lem:matching_pairs}
  Consider two sets of elements, say blue elements $B$ and red elements $R$,
  stored in the leaves of two binary search trees $T^B$ and $T^R$,
  respectively, and let $p \in B$ and $q \in R$, with $q < p$, be leaves in
  trees $T^B$ and $T^R$, respectively. There is a pair of nodes $b \in T^B$ and
  $r \in T^R$, such that
  \begin{itemize}
    \item $p \in P_b$ and $b \in C(T^B,[x',\infty))$, and
    \item $q \in P_r$ and $r \in C(T^R,(-\infty,x])$,
  \end{itemize}
  \noindent
  where $x' = \max P_r$, $x = \min P_b$, and $C(T^S,I)$ denotes the minimal set
  of nodes in $T^S$ whose canonical subsets together represent exactly the
  elements of $S \cap I$.
\end{lemma}

\begin{wrapfigure}[15]{R}{0.3\linewidth}
  \vspace{-1.5\baselineskip}
  \centering
  \includegraphics{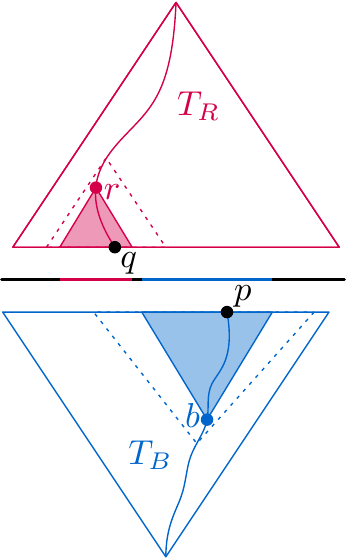}
  \caption{The nodes $b$ and $r$ in the trees $T^B$ and $T^R$.}
  \label{fig:matching_pairs}
\end{wrapfigure}

\noindent\textcolor{darkgray}{\sffamily\bfseries Proof.}
  Let $b$ be the first node on the path from the root of $T^B$ to $p$ such that
  the canonical subset $P_b$ of $b$ is contained in the interval $[q,\infty)$,
  but the canonical subset of the parent of $b$ is not. We define $b$ to be the
  root of $T^B$ if no such node exists. We define $r$ to be the first node on
  the path from the root of $T^R$ to $q$ for which $P_r$ is contained in
  $(-\infty,x]$ but the canonical subset of the parent is not. We again define
  $r$ as the root of $T^R$ if no such node exists. See
  Fig.~\ref{fig:matching_pairs}. Clearly, we now directly have that $r$ is one
  of the nodes whose canonical subsets form $R \cap (-\infty,x]$, and that
  $q \in P_r$ (as $r$ lies on the search path to $q$). It is also easy to see
  that $p \in P_b$, as $b$ lies on the search path to $p$. All that remains is
  to show that $b$ is one of the canonical subsets that together form
  $B \cap [x', \infty)$. This follows from the fact that $q \leq x' < x \leq p$
  ---and thus $P_b$ is indeed a subset of $[x', \infty)$--- and the fact that
  the subset of the parent $v$ of $b$ contains an element smaller than $q$, and
  can thus not be a subset of~$[x', \infty)$.\hfill\qed

\begin{lemma}
  \label{lem:detect_intersecting_squares}
  Let $\square_p$ and $\square_q$, with $p \in D^+(q)$, be the first pair of
  squares to intersect, at some time $t^*$, then there is a pair of nodes
  $w,z$ that have a linking certificate that fails at time $t^*$.
\end{lemma}
\begin{proof}
  Consider the leaves representing $p$ and $q$ in $T^L$ and $T^R$,
  respectively. By Lemma~\ref{lem:matching_pairs} we get that there is a pair
  of nodes $u \in T^L$ and $u' \in T^R$ that, among other properties, have
  $p \in P_u$ and $q \in P_{u'}$. Hence, we can apply
  Lemma~\ref{lem:matching_pairs} again on the associated trees of $u$ and $u'$,
  giving us nodes $v \in T_u$ and $v' \in T_{u'}$ which again have $p \in P_v$
  and $q \in P_{v'}$. Finally, we apply Lemma~\ref{lem:matching_pairs} once
  more on $X^L_v$ and $X^R_{v'}$ giving us nodes $w \in X^L_v$ and
  $z \in X^R_{v'}$ with $p \in P_w$ and $q \in P_z$. In addition, these three
  applications of Lemma~\ref{lem:matching_pairs} give us two points $(x,y,\gamma)$
  and $(x',y',\gamma')$ such that:
  \begin{itemize}
   \item $P_u$ occurs as a canonical subset representing $P \cap ([x',\infty) \times \R^2)$,
   \item $P_v$ occurs as a canonical subset representing $P_u \cap (\R \times [y',\infty)
     \times \R)$, and
   \item $P_w$ occurs as a canonical subset representing $P_v \cap (\R^2 \times
     [\gamma',\infty))$,
   \end{itemize}
   and such that
 \begin{itemize}
   \item $P_{u'}$ occurs as a canonical subset representing $P \cap ((-\infty,x] \times \R^2)$,
   \item $P_{v'}$ occurs as a canonical subset representing $P_{u'} \cap (\R \times (-\infty,y]
     \times \R)$, and
   \item $P_z$ occurs as a canonical subset representing $P_{v'} \cap (\R^2 \times
     (-\infty,\gamma])$.
  \end{itemize}
  Combining these first three facts, and observing that
  $\overline{m}^z = (x',y',\gamma')$ gives us that $P_w$ occurs as a canonical
  subset representing
  $P \cap ([x',\infty) \times [y',\infty) \times [\gamma',\infty)) =
  D^+((x',y',\gamma'))$, and hence $w \in Q^L(\overline{m}^z)=Q^L(z)$. Analogously,
  combining the latter three facts and $\underline{m}^w=(x,y,\gamma)$ gives us
  $z \in Q^R(w)$. Therefore, $w$ and $z$ have a linking certificate. This
  linking certificate involves the leftmost bottom left vertex $\ell_a$ for
  some point $a \in P_w$ and the rightmost top right vertex $r_b$ for some
  point $b \in P_z$. Since $p \in P_w$ and $q \in P_z$, we have that
  $r_q \leq r_b$ and $\ell_a \leq \ell_p$, and thus we detect their
  intersection at time $t^*$.
\end{proof}

\noindent From Lemma~\ref{lem:detect_intersecting_squares} it follows that we can now
detect the first intersection between a pair of squares $\square_p, \square_q$,
with $p \in D^+(q)$. We define an analogous data structure for when
$p \in D^-(q)$. Following Lemma~\ref{lem:intersect}, the kinetic tournaments
will maintain the vertices with minimum and maximum $y$-coordinate for this
case. We then again link up the kinetic tournament nodes in the two trees
appropriately.

% \frank{Observe that all nodes $w$ in a kinetic tournament $X^L_v$ have the same
%   $x$ and $y$ coordinates for their points $\underline{m}^w$. We may be able to
%   exploit that to argue that we need less certificates?
% }

\subparagraph{Space Usage.} Our trees $T^L$ and $T^R$ are range trees in
$\R^3$, and thus use $O(n\log^2 n)$ space. However, it is easy to see that this
is dominated by the space required to store the certificates. For all
$O(n\log^2 n)$ kinetic tournament nodes we store at most $O(\log^3 n)$
certificates (Lemma~\ref{lem:locality}), and thus the total space used by our
data structure is $O(n\log^5 n)$. In Section~\ref{sec:Link_Relations} we will
show that the number of certificates that we maintain is actually only
$O(n (\log n \log\log n)^2)$. This means that our data structure also uses only
$O(n (\log n \log\log n)^2)$ space.

\subsection{Answering Queries}
\label{sub:Queries}
The basic query that our data structure supports is testing if a query square
$\square_q$ currently intersects with a square $\square_p$ in $P$, with
$p \in D^+(q)$. To this end, we simply select the $O(\log^3 n)$ kinetic
tournament nodes whose canonical subsets together represent $D^+(q)$. For each
such a node $w$ we check if the $x$-coordinate of the lower-left vertex
$\ell_p$ stored at that node (which has minimum $x$-coordinate among $L_w$) is
smaller than the $x$-coordinate of $r_q$. If so, the squares intersect. The
correctness of our query algorithm directly follows from
Observation~\ref{obs:intersect}. The total time required for a query is
$O(\log^3 n)$. Similarly, we can test if a given query point $q$ is contained
in a square $\square_p$, with $p \in D^+(q)$. Note that our full data structure
will contain trees analogous to $T^L$ that can be used to check if there is a
square $\square_p \in P$, with $p \in D^-(q)$, or $p$ in one of the other
quadrants defined by $q$, that intersects $\square_q$.

\subsection{Inserting or Deleting a Square}
\label{sub:Updates}

At an insertion or deletion of a square $\square_p$ we proceed in three steps. First, we
update the individual trees $T^L$ and $T^R$, making sure that they once again
represent 3D range trees of all center points in $P$, and that the ternary data structures are, by
themselves, correct kinetic tournaments. For each kinetic tournament node in
$\X^L$ affected by the update, we then query $T^R$ to find a new set of linking
certificates. We update the affected nodes in $\X^R$ analogously. Finally, we
update the global event queue that stores all certificates.
%
% \begin{wrapfigure}[11]{R}{0.45\linewidth}
%   \vspace{-1.75\baselineskip}
%   \centering
%   \includegraphics{rotation}
%   \caption{After a left rotation around an edge $(\mu, \nu)$, the associated data
%     structure $T_\mu$ of node $\mu$ (pink) has to be rebuilt from scratch as its
%     canonical subset has changed. For node $\nu$ we can simply use the old
%     associated data of node $\mu$. No other nodes are affected.}
%   \label{fig:rotation}
% \end{wrapfigure}
%
\begin{lemma}
  \label{lem:updating_the_tree}
  Inserting a square into $T^L$ or deleting a square from $T^L$ takes
  $O(\log^3 n)$ amortized time.
\end{lemma}
\begin{proof}
  We use the following standard procedure for updating the three-level
  \BBalpha{}~trees $T^L$ in $O(\log^3 n)$ amortized time. An update
  (insertion or deletion) in a ternary data structure can easily be handled in
  $O(\log n)$ time. When we insert into or delete an element $x$ in a \BBalpha{}
  tree that has associated data structures, we add or remove the leaf that
  contains $x$, rebalance the tree by rotations, and finally add or remove $x$
  from the associated data structures. When we do a left rotation around an edge
  $(\mu, \nu)$ we have to build a new associated data structure for node $\mu$
  from scratch. See Fig.~\ref{fig:rotation}. Right rotations are handled
  analogously. It is well known that if building the associated data structure at
  node $\mu$ takes $O(|P_\mu|\log^c |P_\mu|)$ time, for some $c \geq 0$, then
  the costs of all rebalancing operations in a sequence of $m$ insertions and
  deletions takes a total of $O(m\log^{c+1}n)$ time, where $n$ is the maximum
  size of the tree at any time~\cite{mehlhorn1984data}. We can build a new
  kinetic tournament $X^L_v$ for node $v$ (using the associated data
  structures at its children) in linear time. Note that this cost excludes
  updating the global event queue. Building a new secondary tree $T_v$,
  including its associated kinetic tournaments, takes $O(|T_v|\log|T_v|)$
  time. It then follows that the cost of our rebalancing operations is at most
  $O(m\log^2 n)$. This is dominated by the total number of nodes created and
  deleted, $O(m\log^3 n)$, during these operations. Hence, we can insert or
  delete a point (square) in $T^L$ in $O(\log^3 n)$ amortized time.
\end{proof}

\begin{figure}[b]
  \centering
  \includegraphics{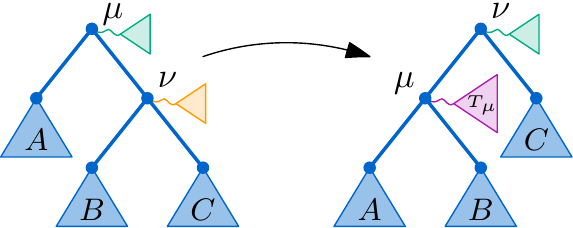}
  \caption{After a left rotation around an edge $(\mu, \nu)$, the associated data
    structure $T_\mu$ of node $\mu$ (pink) has to be rebuilt from scratch as its
    canonical subset has changed. For node $\nu$ we can simply use the old
    associated data of node $\mu$. No other nodes are affected.}
  \label{fig:rotation}
\end{figure}

Analogous to Lemma~\ref{lem:updating_the_tree} we can update $T^R$ in $O(\log^3
n)$ amortized time. Next, we update the linking certificates. We say that a
kinetic tournament node $w$ in $T^L$ is \emph{affected by} an update if (i) the
update added or removed a leaf node in the subtree rooted at $w$, (ii) node $w$
was involved in a tree rotation, or (iii) $w$ occurs in a newly built
associated tree $X^L_v$ (for some node $v$). Let $\X^L_i$ denote the set of
nodes affected by update $i$. Analogously, we define the set of nodes $\X^R_i$
of $T^R$ affected by the update. For each node $w \in \X^L_i$, we query $T^R$ to
find the set of $O(\log^3 n)$ nodes whose canonical subsets represent $Q^R(w)$.
For each node $z$ in this set, we test if we have to add a linking certificate
between $w$ and $z$. As we show next, this takes constant time for each node
$z$, and thus $O(\sum_i |\X^L_i| \log^3 n)$ time in total, for all nodes $w$. We
update the linking certificates for all nodes in $\X^R_i$ analogously.

We have to add a link between a node $z \in Q^R(w)$ and $w$ if and only if we
also have $w \in Q^L(z)$. We test this as follows. Let $v$ be the node whose
associated tree $X^L_v$ contains $w$, and let $u$ be the node in $T^L$ whose
associated tree contains $v$. We have that $w \in Q^L(z)$ if and only if
$u \in C(T^L,[\overline{m}^z_x,\infty))$,
$v \in C(T_u,[\overline{m}^z_y,\infty))$, and
$w \in C(X^L_v,[\overline{m}^z_\gamma,\infty))$. We can test each of these
conditions in constant time:
\begin{observation}
  \label{obs:ancestor_testing}
  Let $q$ be a query point in $\R^1$, let $w$ be a node in a binary search tree
  $T$, and let $x_p = \min P_p$ of the parent $p$ of $w$ in $T$, or
  $x_p = -\infty$ if no such node exists. We have that $w \in C(T,[q,\infty))$ if
  and only if $q \leq \min P_w$ and $q > x_p$. % This can be tested in constant
  % time.
\end{observation}
Finally, we delete all certificates involving no longer existing nodes from our
global event queue, and replace them by all newly created certificates. This
takes $O(\log n)$ time per certificate. We charge the cost of deleting a
certificate to when it gets created. Since every node $w$ affected creates
at most $O(\log^3 n)$ new certificates, all that remains is to bound the total number
of affected nodes. We can show this using basically the same argument as we
used to bound the update time. This leads to the following result.

% To bound the total number of affected nodes we basically use the same argument
% used to bound the update time. This leads to the following result.

\begin{lemma}
  \label{lem:update}
  Inserting a disjoint square into $P$, or deleting a square from $P$ takes
  $O(\log^7 n)$ amortized time.
\end{lemma}
\begin{proof}
  An update visits at most $O(\log^3 n)$ nodes itself (i.e.~leaf nodes and
  nodes on the search path). All other affected nodes occur as newly built
  trees due to rebalancing operations. As in Lemma~\ref{lem:updating_the_tree},
  the total number of nodes created due to rotations in a sequence of $m$
  updates is $O(m\log^2 n)$. It follows that the total number of affected nodes
  in such a sequence is $O(m\log^3 n)$. Therefore, we create $O(m\log^6 n)$
  linking certificates in total, and we can compute them in $O(m\log^6 n)$
  time. Updating the event global queue therefore takes $O(m\log^7 n)$ time.
\end{proof}

\subsection{Running the Simulation}
\label{sub:Running_the_Simulation}

All that remains is to analyze the number of events processed. We show that
in a sequence of $m$ operations, our data structure processes at most
$O(m\alpha(n)\log^3 n)$ events. This leads to the following result.
\begin{theorem}
  \label{thm:data_structure}
  We can maintain a set $P$ of $n$ disjoint growing squares in a fully dynamic
  data structure such that we can detect the first time that a square
  $\square_q$ intersects with a square $\square_p$, with $p \in D^+(q)$. Our
  data structure uses $O(n (\log n \log\log n)^2)$ space, supports updates in $O(\log^7 n)$
  amortized time, and queries in $O(\log^3 n)$ time. For a sequence of $m$
  operations, the structure processes a total of $O(m\alpha(n)\log^3 n)$ events
  in a total of $O(m\alpha(n)\log^7 n)$ time.
\end{theorem}
\begin{proof}
  We argued the bounds on the space, the query, and the update times
  before. All that remains is to bound the number of events processed, and the
  time to do so.

  We start by the observation that each failure of a linking
  certificate produces an intersection, and thus a subsequent update. It thus
  follows that the number of such events is at most $m$.

  To bound the number of events created by the tournament trees we extend the
  argument of Agarwal \etal~\cite{agarwal2008kinetictournament}. For any
  kinetic tournament node $w$ in $T^L$, the minimum $x$-coordinate corresponds
  to a lower envelope of line-segments in the $t,x$-space. This envelope has
  complexity $O(|P^*_w|\alpha(|P^*_w|))=O(|P^*_w|\alpha(n))$, where $P^*_w$ is
  the multiset of points that ever occur in $P_w$, i.e.~that are stored in a
  leaf of the subtree rooted at $w$ at some time $t$. Hence, the number of
  tournament events involving node $w$ is also at most
  $O(|P^*_w|\alpha(n))$. It then follows that the total number of events is
  proportional to the size of these sets $P^*_w$, over all $w$ in our tree. As
  in Lemma~\ref{lem:updating_the_tree}, every update directly contributes one
  point to $O(\log^3 n)$ nodes. The remaining contribution is due to
  rebalancing operations, and this cost is again bounded by $O(m\log^2
  n)$. Thus, the total number of events processed is $O(m\alpha(n)\log^3 n)$.

  At every event, we have to update the $O(\log^3 n)$ linking certificates of
  $w$. This can be done in $O(\log^4 n)$ time (including the time to update the
  global event queue). Thus, the total time for processing all kinetic tournament
  events in $T^L$ is
  $O(m\alpha(n)\log^7 n)$. The analysis for the kinetic tournament nodes $z$ in
  $T^R$ is analogous.
\end{proof}
To simulate the process of growing the squares in $P$, we now maintain eight
copies of the data structure from Theorem~\ref{thm:data_structure}: two data
structures for each quadrant (one for $D^+$, the other for $D^-$). We thus
obtain the following result.
\begin{theorem}
  \label{thm:main_datastructure}
  We can maintain a set $P$ of $n$ disjoint growing squares in a fully dynamic
  data structure such that we can detect the first time that two squares in $P$
  intersect. Our data structure uses $O(n (\log n \log\log n)^2)$ space, supports updates in
  $O(\log^7 n)$ amortized time, and queries in $O(\log^3 n)$ time. For a
  sequence of $m$ operations, the structure processes $O(m\alpha(n)\log^3 n)$
  events in a total of $O(m\alpha(n)\log^7 n)$ time.
\end{theorem}
And thus we obtain the following solution to the agglomerative glyph
clustering problem.
\begin{theorem}
  \label{thm:agglomerative_glyph_clustering}
  Given a set of $n$ initial square glyphs $P$, we can compute an agglomerative
  clustering of the squares in $P$ in $O(n\alpha(n)\log^7 n)$ time using
  $O(n (\log n \log\log n)^2)$ space.
\end{theorem}

\section{Efficient Representation of Dominance Relations}
\label{sec:Link_Relations}

The linking certificates of our data structure actually comprise an efficient representation of all dominance relations between two point sets. We therefore think that this representation, and in particular the tighter analysis in this section, is of independent interest.

Let $R$ and $B$ be two point sets in $\R^d$ with $|R| = n$ and $|B| = m$, and let $T^R$ and $T^B$ be range trees built on $R$ and $B$, respectively. We assume that each layer of $T^R$ and $T^B$ consists of a \BBalpha{}-tree, although similar analyses can be performed for other types of balanced binary search trees. By definition, every node $u$ on the lowest layer of $T^R$ or $T^B$ has an associated $d$-dimensional range $Q_u$ (the hyper-box, not the subset of points). For a node $u \in T^R$, we consider the subset of points in $B$ that dominate all points in $Q_u$, which can be comprised of $O(\log^d m)$ canonical subsets of $B$, represented by nodes in $T^B$. Similarly, for a node $v \in T^B$, we consider the subset of points in $R$ that are dominated by all points in $Q_v$, which can be comprised of $O(\log^d n)$ canonical subsets of $R$, represented by nodes in $T^R$. We now link a node $u \in T^R$ and a node $v \in T^B$ if and only if $v$ represents such a canonical subset for $u$ and vice versa. By repeatedly applying Lemma~\ref{lem:matching_pairs} for each dimension, it can easily be shown that these links represent all dominance relations between $R$ and $B$.

As a $d$-dimensional range tree consists of $O(n \log^{d-1} n)$ nodes, a trivial bound on the number of links is $O(m \log^{2d - 1} n)$ (assuming $n \geq m$). Below we show that the number of links can be bounded by $O(n (\log n \log\log n)^{d-1})$. We first consider the case for $d=1$.

%We discuss below how the number of linking certificates is a major factor of the
%running time of updates to our data structure. It is therefore important to keep
%the number of linking certificates as small as possible. The explanation above
%is good for intuition, but implies many certificates covering essentially the
%same information. We will now discuss how to streamline this and significantly
%reduce the number of linking certificates. First we present the idea in 1D,
%after that we show how to generalize it to higher dimensions. At that point it
%can directly be used to reduce the number of linking certificates.
%
%Let $R$ be a set of $n$ ``red'' points in $\R^d$, and let $B$ be a set of $m$
%``blue'' points in $\R^d$. We store $R$ and $B$ in range trees implemented
%using multi-layered \BBalpha{}-trees $T_R$ and $T_B$, respectively. With each
%node $u$ in these trees we associate a \emph{query range} $Q_u$.
%\frank{define the corner points $\underline{m}$ and $\overline{m}$ of these ranges}

\subsection{Analyzing the Number of Links in 1D}
Let $R$ and $B$ be point sets in $\R$ with $|R| = n$, $|B| = m$, and $n \geq m$. Now, every associated range of a node $u$ in $T^R$ or $T^B$ is an interval $I_u$. We can extend the interval to infinity in one direction; to the left for $u \in T^R$, and to the right for $u \in T^B$. For analysis purposes we construct another range tree $T$ on $R \cup B$, where $T$ is not a \BBalpha{}-tree, but instead a perfectly balanced tree with height $\lceil\log(n + m)\rceil$. For convenience we assume that the associated intervals of $T$ are slightly expanded so that all points in $R \cup B$ are always interior to the associated intervals. We associate a node $u$ in $T^R$ or $T^B$ with a node $v$ in $T$ if the endpoint of $I_u$ is contained in the associated interval $I_v$ of $v$.
%
%As we did for Lemma~\ref{lem:matching_pairs}, let us consider two sets of
%elements, $n$ red elements $R$ and $m$ blue elements $B$, stored in the leaves
%of \BBalpha{}~trees $T^R$ and $T^B$, respectively. With every node $u$ in $T^R$
%or $T^B$ we associate the appropriate corresponding interval $I_u$, which
%extends to infinity in one direction; to the left for red elements, to the right
%for blue elements. Consider a 1D range tree $T$ on the endpoints of all
%intervals, $R \cup B$. Note that we use $T$ only in the analysis, there is no
%need to build $T$ explicitly. Every node of $T$ will have an associated
%interval. For convenience, let any internal node store not the maximum value
%contained in its left subtree, but the average of that value and the minimum
%value contained in its right subtree. This ensures that every endpoint of an
%interval (of $R$ or $B$) will be internal to intervals associated with nodes of
%$T$. Since there are $O(n + m)$ intervals with $O(n + m)$ endpoints, $T$ has
%height $O(\log(n + m))$.
%
%We associate a node $u$ in $T^R$ or $T^B$ with a node $v$ in $T$ if the endpoint
%of $I_u$ is contained in the associated interval $I_v$ of $v$.
%
\begin{observation}
  \label{obs:once_per_level}
  Every node of $T^R$ or $T^B$ is associated with at most one node per level
  of~$T$.
\end{observation}
For two intervals $I_u = (-\infty, a]$ and $I_v = [b, \infty)$, corresponding to
a node $u \in T^R$ and a node $v \in T^B$, let $[a, b]$ be the
\emph{spanning interval} of $u$ and $v$. We now want to charge spanning
intervals of links to nodes of $T$.
%A node $u \in T^R$ and $v \in T^B$ should be linked
%if and only if, using the notation introduced in Lemma~\ref{lem:matching_pairs},
%$u \in C(T^R, (-\infty, v]) \wedge v \in C(T^B, [u, \infty))$.
We charge a spanning interval $I_{uv} = [a, b]$ to a node $w$ of $T$ if and only if $[a, b]$
is a subset of $I_w$, and $[a, b]$ is cut by the splitting coordinate of $w$.
Clearly, every spanning interval can be charged to exactly one node of $T$.

Now, for a node $u$ of $T$, let $h_R(u)$ be the height of the highest node of
$T^R$ associated with $u$, and let $h_B(u)$ be the height of the highest node of
$T^B$ associated with $u$.
\begin{lemma}
  \label{lem:num_spanning_intervals_node}
  The number of spanning intervals charged to a node $u$ of $T$ is $O(h_R(u) \cdot h_B(u))$.
\end{lemma}
\begin{proof}
Let $x$ be the splitting coordinate of $u$ and let $r \in T^R$ and $b \in T^B$ form a spanning interval that is charged to $u$. We claim that, using the notation introduced in Lemma~\ref{lem:matching_pairs}, $r \in C(T^R, (-\infty, x])$ (and symmetrically, $b \in C(T^B, [x, \infty))$). Let $I_b = [x', \infty)$ be the associated interval of $b$, where $x' > x$. By definition, $r \in C(T^R, (-\infty, x'])$. If $r \notin C(T^R, (-\infty, x])$, then the right endpoint of $I_r$ must lie between $x$ and $x'$. But then the spanning interval of $r$ and $b$ would not be charged to $u$. As a result, we can only charge spanning intervals between $h_R(u)$ nodes of $T^R$ and $h_B(u)$ nodes of $T^B$, of which there are at most $O(h_R(u) \cdot h_B(u))$.
\end{proof}
Using Lemma~\ref{lem:num_spanning_intervals_node}, we count the total number of
charged spanning intervals and hence, links between $T^R$ and $T^B$. We refer to
this number as $\numLinks{T^R, T^B}$. This is simply $\sum_{u \in T} O(h_R(u)
\cdot h_B(u)) \leq \sum_{u \in T} O(h_R(u)^2 + h_B(u)^2)$. We can split the sum and assume w.l.o.g. that
$\numLinks{T^R, T^B} \leq 2\sum_{u \in T} O(h_R(u)^2)$. Rewriting the sum based on heights in $T^R$ gives
\[
  \numLinks{T^R, T^B} \leq \sum_{h_R\,=\,0}^{\textit{height}(T^R)} n_T(h_R)
    \cdot O(h_R^2),
\]
where $n_T(h_R)$ is the number of nodes of $T$ that have a node of height $h_R$
associated with it.

To bound $n_T(h)$ we use Observation~\ref{obs:once_per_level} and the fact that
$T^R$ is a \BBalpha{}~tree. Let $c = \frac{1}{1 - \alpha}$, then we get that
$\textit{height}(T^R) \leq \log_c(n)$ from properties of \BBalpha{}~trees.
Therefore, the number of nodes in $T^R$ that have height $h$ is at most
$O(\frac{n}{c^h})$.
%Since every node of $T^R$ can be associated with
%$O(\log(n + m))$ nodes of $T$, we could use the naive bound of
%$n_T(h) = \min\!\left(n + m, \frac{(n + m) \log(n + m)}{c^h}\right)$, but
%unfortunately this is not strong enough.
%
\begin{lemma}\label{lem:nodes-at-height}
  $n_T(h) = O\left(\frac{(n + m) h}{c^h}\right)$.
\end{lemma}
\begin{proof}
  As argued, there are at most $O(n/c^h)$ nodes in $T^R$ of height $h$. Consider cutting the tree $T$ at level $\log(n/c^h)$. This results in a top tree of size $O(n/c^h)$, and $O(n/c^h)$ bottom trees. Clearly, the top tree contributes at most its size to $n_T(h)$. All bottom trees have height at most $\lceil\log(n + m)\rceil - \log(n/c^h) = O(\log(c^h) + \log(1 + m/n)) = O(h + m/n)$.
  Every node in $T^R$ of height $h$ can, in the worst case, be associated with one distinct node per level in the bottom trees by Observation~\ref{obs:once_per_level}. Hence, the bottom trees contribute at most $O(n (h + m/n) / c^h) = O((n h + m) / c^h) = O((n + m) h / c^h)$ to $n_T(h)$.
\end{proof}

Using this bound on $n_T(h)$ in the sum we previously obtained gives:
\[
  \numLinks{T^R, T^B} \leq \sum_{h_R = 0}^{height(T^R)} O\!\left(\frac{(n + m)
  h_R^3}{c^{h_R}}\right) \leq O(n + m) \sum_{h = 0}^{\infty} \frac{h^3}{c^h} = O(n + m).
\]
Where indeed, $\sum_{h = 0}^{\infty} \frac{h^3}{c^h} = O(1)$ because $c >
1$. Thus, we conclude:

\begin{theorem}\label{thm:link-count-1d}
  The number of links between two $1$-dimensional range trees $T^R$ and $T^B$ containing $n$ and $m$
  points, respectively, is bounded by $O(n + m)$.
\end{theorem}

\subsection{Extending to Higher Dimensions}
We now extend the bound to $d$ dimensions. The idea is very simple. We first determine the links for the top-layer of the range trees. This results in links between associated range trees of $d-1$ dimensions (see Fig.~\ref{fig:trees-2d}). We then determine the links within the linked associated trees, which number can be bounded by induction on $d$.

\begin{theorem}\label{thm:link-count-nd}
  The number of links between two $d$-dimensional range trees $T^R$ and $T^B$ containing $n$ and $m$ ($n \geq m$)
  points, respectively, is bounded by $O(n (\log n \log\log n)^{d-1})$.
\end{theorem}
\begin{proof}
We show by induction on $d$ that the number of links is bounded by the minimum of $O(n (\log n \log\log n)^{d-1})$ and $O(m \log^{2 d - 1} n)$. The second bound is simply the trivial bound given at the start of Section~\ref{sec:Link_Relations}. The base case for $d=1$ is provided by Theorem~\ref{thm:link-count-1d}. Now consider the case for $d > 1$. We first determine the links for the top-layer of $T^R$ and $T^B$. Now consider the links between an associated tree $T_u$ in $T^R$ containing $k$ points and other associated trees $T_0, \ldots, T_r$ that contain at most $k$ points. Since $T_u$ can be linked with only one associated tree per level, and because both range trees use \BBalpha{}~trees, the number of points $m_0, \ldots, m_r$ in $T_0, \ldots, T_r$ satisfy $m_i \leq k/c^i$ ($0 \leq i \leq r$) where $c = \frac{1}{1 - \alpha}$. By induction, the number of links between $T_u$ and $T_i$ is bounded by the minimum of $O(k (\log n \log\log n)^{d-2})$ and $O(m_i \log^{2d - 3} n)$. Now let $i^* = \log_c(\log^{d-1} n) = O(\log\log n)$. Then, for $i \geq i^*$, we get that $O(m_i \log^{2d - 3} n) = O(k \log^{d-2} n)$. Since the sizes of the associated trees decrease geometrically, the total number of links between $T_u$ and $T_i$ for $i \geq i^*$ is bounded by $O(k \log^{d-2} n)$. The links with the remaining trees can be bounded by $O(k \log^{d-2} n (\log\log n)^{d-1})$. Finally note that the top-layer of each range tree has $O(\log n)$ levels, and that each level contains $n$ points in total. Thus, we obtain $O(n \log^{d-1} n (\log\log n)^{d-1})$ links in total. The remaining links for which the associated tree in $T^B$ is larger than in $T^R$ can be bounded in the same way. %This concludes the proof by induction.
\end{proof}

%We extend the above argument to range trees in $\R^2$, and show that the number
%of links is $O(n\log^2 n)$. Our argument easily extends to $\R^d$, for any
%constant dimension $d$, yielding a bound of $O(n\log^{2d-2}n)$. Consider two
%\BBalpha{}~trees $T_R$ and $T_B$, both with two layers. These trees are built
%on the $x$-coordinates of two sets of points $R$ and $B$. All primary nodes $u$
%in both trees have an associated \BBalpha{}~tree storing the canonical subset
%of $u$, but sorted on $y$-coordinate. Recall that only nodes on the lowest
%level are linked, so we now count the links within the secondary level to
%obtain the total number of links $\numLinks{T_R, T_B}$. See
%Fig.~\ref{fig:trees-2d} for an illustration.
%
%Consider two \BBalpha{}~trees $T^R$ and $T^B$, both with two layers. These trees
%are built on the $x$-coordinates of two sets of points $R$ ($|R| = n$) and $B$
%($|B| = m$). All primary nodes $u$ in both trees have an associated
%\BBalpha{}~tree storing the canonical subset of $u$, but sorted on
%$y$-coordinate. We refer to such associated trees as $T^R_u$ or $T^B_u$. Recall
%that only nodes on the lowest level are linked, so we now count the links within
%the secondary level to obtain the total number of links $\numLinks{T^R, T^B}$.
%See Fig.~\ref{fig:trees-2d} for an illustration.

\begin{figure}[tb]
  \centering
  \includegraphics{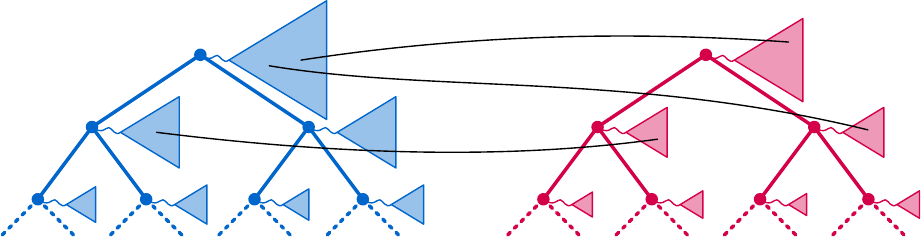}
  \caption{Two layered trees with two layers, and the links between them
    (sketched in black). We are interested in bounding the number of such links.}
  \label{fig:trees-2d}
\end{figure}

\noindent It follows from Theorem~\ref{thm:link-count-nd} that our data structure from
Section~\ref{sec:DataStructure} actually maintains only $O(n (\log n \log\log n)^2)$
certificates. This directly implies that the space usage is only $O(n (\log n \log\log n)^2)$
as well.

\section{Conclusion and Future Work}
\label{sec:Concluding_Remarks}

We presented an efficient fully dynamic data structure for maintaining a set of
disjoint growing squares. This leads to an efficient algorithm for
agglomerative glyph clustering. The main future challenge is to improve the
analysis of the running time. Our analysis from
Section~\ref{sec:Link_Relations} shows that at any time, we need only few
linking certificates. However, we would like to bound the total number of
linking certificates used throughout the entire sequence of operations. An
interesting question is if we can extend our argument to this case. This may
also lead to a more efficient algorithm for maintaining the linking
certificates during updates.
% , as in our current approach we simply test all $O(\log^3 n)$
% candidate nodes in a canonical subset for each node affected by the update.

%%
%% Bibliography
%%

%% Either use bibtex (recommended),

\bibliography{../agglomerative-clustering-growing-squares}

%% .. or use the thebibliography environment explicitely

\end{document}